\documentclass[10pt,journal,twocolumn]{IEEEtran}
\makeatletter
\def\ps@headings{%
\def\@oddhead{\mbox{}\scriptsize\rightmark \hfil \thepage}%
\def\@evenhead{\scriptsize\thepage \hfil \leftmark\mbox{}}%
\def\@oddfoot{}%
\def\@evenfoot{}}
\makeatother

\usepackage{mathrsfs}
\usepackage{amsmath}
\usepackage{amssymb}
\usepackage{subfigure}
\usepackage{bbm}
\usepackage{graphics}
\usepackage{psfrag}
\usepackage{epsfig}
\usepackage[verbose]{cite}
\usepackage{algorithm}
\usepackage{algorithmic}
\usepackage{cite}

\newcommand{\ora}{\overrightarrow}
\newcommand{\G}{\ora{G}}
\newcommand{\E}{\ora{E}}
\newcommand{\e}{\ora{e}}
\newcommand{\ve}{\varepsilon}

\newcommand{\T}{\mathbbm{T}}

\newcommand{\mca}{\mathcal{A}}
\newcommand{\mcf}{\mathcal{F}}
\newcommand{\mcb}{\mathcal{B}}
\newcommand{\mcd}{\mathcal{D}}
\newcommand{\mcg}{\mathcal{G}}

\newcommand{\mcv}{\mathcal{V}}

\newcommand{\mcc}{\mathcal{C}}

\newcommand{\ol}{\overline}

\newtheorem{lem}{Lemma}

\newtheorem{thm}{Theorem}
\newtheorem{prf}{Proof}

\newtheorem{define}{Definition}

\title{Exact Modeling of the Performance of Random Linear Network Coding in Finite-buffer Networks}

\author{Nima Torabkhani$^\dag$,
     Badri N. Vellambi$^{\ddag}$,
	  Ahmad Beirami$^\dag$,
     Faramarz Fekri$^\dag$\\
     $^\dag$ School of Electrical and Computer Engineering,
    Georgia Institute of Technology\\
     $^{\ddag}$ Institute for Telecommunications Research,
    University of South Australia\\
     E-mail: \{nima, beirami, fekri\}@ece.gatech.edu, badri.vellambi@unisa.edu.au
\vspace{-.25in}
\thanks{This material is based upon work supported by the National Science Foundation under Grant No. CCF-0914630, and the Australian Research
Council under ARC Discovery Grants DP0880223 and DP1094571.}
}

\allowdisplaybreaks[4]
\begin{document}
\maketitle
\pagestyle{empty}
\thispagestyle{empty}

\begin{abstract}
In this paper, we present an exact model for the analysis of the performance of Random Linear Network Coding (RLNC) in wired erasure networks with finite buffers. In such networks, packets are delayed due to either random link erasures or blocking by full buffers. We assert that because of RLNC, the content of buffers have dependencies which cannot be captured directly using the classical queueing theoretical models. We model the performance of the network using Markov chains by a careful derivation of the buffer occupancy states and their transition rules. We verify by simulations that the proposed framework results in an accurate measure of the network throughput offered by RLNC. Further, we introduce a class of acyclic networks for which the number of state variables is significantly reduced.
\end{abstract}

\section{Introduction}\label{FB-intro}

It is well-known that linear network codes achieve the min-cut capacity of networks for unicast applications~\cite{YeungNCJrnl}. In fact, random linear codes over large Galois fields suffice to achieve the min-cut capacity~\cite{Koetter2003}. \emph{Random linear network coding} (RLNC) has been shown to improve the performance in distributed settings with time-varying network parameters. In these networks, a distributed and packetized network coding scheme, where each node stores received packets and forwards random linear combinations of the stored packets when required, was introduced in~\cite{PNC2003}. As a result, for a network of nodes with no buffer limitations, all arriving packets at a node are stored, and then used to generate new packets to send. Hence, there is no information loss. However, in this case, upon reception of a packet, a node has to determine whether or not the incoming packet is in the linear span of its previously stored packets. Further, for generating every coded packet, all stored packets need to be accessed. It is therefore desirable to have limited buffer sizes, since it limits the complexity of storage and coded packet generation process. Further, using small buffers at relay nodes simplifies practical issues such as on-chip board space and memory-access latency as well as reducing the average packet delay~\cite{Appenzeller04CCR, Yashar06routerswith}.

The problem of computing capacity and designing efficient coding schemes for erasure networks has been widely studied in the absence of buffer constraints\cite{AmirDana:capacity, PakzadF05, YeungNCJrnl}. The limitations posed by finite buffers were considered by~\cite{NetCod:Lun}, specifically in a simple two-hop line network. Inspired by this work, in~\cite{Vellambi2010}, the authors present a Markov-chain-based approach to model the dynamics of the system and the packet occupancy of every intermediate node to approximate the performance parameters (throughput and latency) of a multi-hop line network with lossy links. Several challenges arise when extending the study from a single intermediate node to a multi-hop line network. Results from~\cite{Vellambi2010} were extended to other communication scenarios, such as block-based random linear coding for line networks~\cite{Asilomar-2010}, and general wired networks with lossless feedback and random routing~\cite{ITW-2010}. However, the main challenge of modeling the evolution of \emph{buffer occupancy} or \emph{innovativeness of buffer contents} in general network topologies when RLNC is used, was not addressed in these works.

The queueing theory framework for lossy networks with finite buffers of~\cite{Tayfur_QT1,Tayfur_QT2} attempts to model the packets of the network as customers, the delay due to packet loss over links as service times in the nodes, and the buffer size at intermediate nodes as the maximum queue size. However, this packet-customer equivalence fails to accurately model RLNC in general network topologies. This is due to the possibility of packet replication at intermediate nodes, or more generally, the potential correlation in the contents of the buffers of various intermediate nodes. This correlation or dependency between contents of the buffers cannot be captured directly in the customer-server based queueing model.

In this paper, our objective is to study the relation between throughput of RLNC and the buffer sizes of intermediate nodes in the small buffer regime. The first and the key step in our approach is to derive using algebraic tools the state of the buffers using which the dynamics of the network can be completely characterized. We then derive the state update rules for each transmission in the network. Finally, using the developed state space and update rules, we obtain the throughput of the network using Monte Carlo simulations and compare the results to the actual packetized implementation of RLNC. We believe the proposed modeling framework is a significant step towards developing a theoretical framework for computing the throughput capacity and the packet delay distribution in general finite-buffer wired networks.

This paper is organized as follows. First, we present a formal definition of the problem and the challenges in Section~\ref{FB-sec1}. Next, we investigate the tools and steps for modeling the buffer states in Section~\ref{FB-sec2}. We then introduce in Section~\ref{FB-sec3}, a general class of networks for which the complexity of our modeling is significantly lesser. Finally, Section~\ref{FB-sec4} presents our model validation results using simulations. Conclusions are summarized in Section~\ref{FB-sec5}.

\section{Problem Setup and Challenges}\label{FB-sec1}
Throughout this work, we model the network by an acyclic directed graph $\G(V,\E)$, where packets can be transmitted over a link $\e=(u,v)$ only from the node $u$ to $v$. The system is analyzed using a discrete-time model; each node can transmit at most one packet over a link in an epoch. The loss process on each link is assumed to be memoryless, i.e., packets transmitted on a link $\ora{e}=(u,v)\in\ora{E}$ are lost randomly with a probability of $\ve_{\ora{e}}=\ve_{(u,v)}$. Note that the erasures are due to the quality of links (e.g., noise, interference) and do not represent packet blocking due to finiteness of the buffers. Further, the packet loss processes on different links are assumed to be independent. Each node $v\in V$ has a buffer size of $m_v$ packets with each packet having a fixed size. Source and destination are assumed to be able to store an infinitude of packets. Throughout this paper, node $s$ and node $d$ represent the source and destination nodes, resp. Also, for any $x\in[0,1]$, $\ol{x}\triangleq 1-x$. The unicast information-theoretic throughput is also defined as the expected rate (in packets/epoch) at which information packets are transferred from the source to the destination when the network is in steady-state. In other words, if $\tau_k$ is the time it takes for $k$ information packets to be transmitted to the destination, the throughput capacity is given by
\begin{equation}
\mcc(\ora{G})=\lim_{k\to\infty}{(\tau_k)^{-1}}k.
\label{TPTdef}
\end{equation}

There are two key challenges in finite-buffer networks. The first challenge is the choice of optimal buffer management strategy, which also depends on the routing/coding scheme that is in use. Due to losses on links, and finiteness of buffers, transmission of a packet by a node $u$ on $\ora{e}=(u,v)$ does not guarantee successful reception by the node $v$. Thus, in the absence of any feedback, a node $u$ does not know if it can delete a packet from its buffer to make room for its next incoming packet. Further, it is also unclear if transmitting a packet via several parallel paths will increase the throughput. The second challenge is due to the possible replication of packets in the network. Hence, it is neither possible to model the system dynamics by a simple queueing model where packets are customers and the buffers as queue sizes, nor is it feasible to treat the packets as flows in the network.

{Random Linear Network Coding} (RLNC) attractively bypasses these two challenges. It eliminates the need for a feedback strategy to delete stored packets because the physical act of storing a packet becomes immaterial. It also eliminates the need for active replication by allowing transmitted/stored packets to be treated as elements of an abstract vector space. This makes RLNC a favorable choice for practical schemes in finite-buffer scenarios.

We consider the following packet-coding scheme introduced in~\cite{NetCod:Lun}, which is a finite-buffer adaptation of RLNC. In this scheme, at each epoch, random linear coding is used for both the packet generation and storage by intermediate nodes. As an example, consider a node $u$ of buffer size $m_u$. At a given epoch, $u$ generates an encoded packet by performing a random linear combinations of $m_u$ stored data packets (over a sufficiently large Galois field\footnote{The size of the Galois field needs to be sufficiently large to increase the chance of innovativeness of the coded packet.} $\mathbb{F}_q$), and transmits the coded packet on an outgoing link. For storage, when a packet successfully arrives at a node $v$, the node multiplies the received packet by a random vector chosen uniformly from $\mathbb{F}_q^{m_v}$, and adds the resultant vector components to each of the present buffer contents.

Therefore, using RLNC, after just a single packet reception, the entire buffer becomes physically full with multiples of the received packet. Thus, even though the buffer of the node $u$ is almost always physically full, the number of stored packets that is innovative w.r.t any other subset of nodes can vary from 0 to $m_u$. As an example, suppose that two nodes $a$ and $b$ receive/store two packets each generated from three original packets from a relay $c$. In this case, $a$ and $b$ will have two innovative packets each for the destination. Now, suppose $a$ delivers a packet to the destination. Then, $b$ still contains two innovative packets for the destination. However, if $a$ delivers another packet to the destination, $b$ will only have one innovative packet for the destination, since both nodes together originally possessed only three innovative packets for the destination. In this example, the challenges of tracking the number of innovative packets and the interdependency between buffer contents gets compounded further as the packets from $a$ and $b$ are propagated to the other intermediate nodes. This interdependency between buffer contents signals the need for a novel notion of \emph{occupancy} to track the number of innovative packets each node has for the destination, and consequently, to determine the throughput capacity of the network. This notion will be formalized in the following section.

The main motivating factor to develop a theoretical model for these networks is to understand the throughput capacity under RLNC. In order to measure the throughput of RLNC in these networks, one option is to perform a Monte Carlo simulation where encoded packets are generated using coefficients in a large finite field $\mathbb{F}_q$, and buffer updates are performed upon each successful reception. This is a significantly time-consuming simulation due to large field operations. A theoretical model that tracks buffer dynamics based on occupancy of buffers will be a simpler alternate means. As we will see, the developed model provides a more efficient way of measuring the performance of finite-buffer networks. Additionally, it provides us with intuitive insights on the dynamics of buffer updates, which is a major step towards computing performance metrics for such networks, and analyzing their key trade-offs.

\section{Exact Modeling of Finite-buffer RLNC}
\label{FB-sec2}

Here, we introduce the tools and steps that enable us to track changes in the buffer contents of nodes.

To identify the throughput as defined in (\ref{TPTdef}), we assume that the source possesses a sufficiently large block of packets that has to be transmitted to the destination. The first aim is to formalize the notion of buffer occupancy by investigating the dimension of the span of the stored packets in the buffers. Let $\{T_1,T_2,\ldots,T_k\}$ be the original information packets at the source. Let $[n]\triangleq\{1,2,\ldots,n\}$ denote the set of all intermediate nodes, where $n=|V|-2$. Let $P_{i,j}(t)$ be the packet contained in buffer slot $j$ of relay $i$ at time epoch $t$, where $P_{i,j}(t)=\sum_{l=1}^{k}a_{i,j,l}T_l$, $i\in[n]$, $j\in[m_i]$, and $a_{i,j,l}$ is a coefficient in the chosen Galois field $\mathbb{F}_q$. Let $\mcv(S)(t)\triangleq \text{span}\{P_{i,j}(t)|~ j\in[m_i], i\in S\}$ for all $S\subseteq[n]$. To simplify the notations, we will drop the reference to time in $\mcv(S)(t)$ by using $\mcv(S)$. Also, we define $S^c\triangleq [n]\setminus S$.

\begin{define}
For any two subsets of the intermediate nodes $S,S'\subseteq[n]$, we define the \emph{innovativeness} of $S$ w.r.t. $S'$ at time instant $t$ as:
\begin{align}
I_{S\to S'}= \dim\big(\mcv(S)\big) - \dim\big(\mcv(S)\cap \mcv(S')\big).\label{InnovDef}
\end{align}
\end{define}
In other words, $I_{S\to S'}$ gives the number of innovative packets that buffer contents of nodes in $S$ can generate which cannot be generated by the contents of the buffers of nodes in $S'$.
\begin{define}
The occupancy vector $\{b_S\}_{S\subseteq[n]}$ of the network is defined~\footnote{The precise definition of the occupancy vector must consider the packets that have already reached $\{d\}$ by using $b_S\triangleq \dim(\mcv(S)) - \dim(\mcv(S)\cap \mcv(S^c\cup \{d\}))$. However, the inclusion of $\{d\}$ affects update rules only when dealing with the destination. For simplicity, the equivalent definition without the inclusion of $\{d\}$ is used in all cases not involving the destination.} to be
\begin{align}
b_S\triangleq \dim\big(\mcv(S)\big) - \dim\big(\mcv(S)\cap \mcv(S^c)\big),\,\, S\subseteq [n].
\end{align}
\label{OccDef}
\end{define}
\vspace{-.15in}
The following lemma shows that the knowledge of occupancy vector $\{b_S\}_{S\subseteq[n]}$ is equivalent to knowing the innovativeness of any subset of the relay nodes w.r.t. any other subset. This result significantly reduces the number of state space variables.
\begin{lem}
For $S,S'\subseteq[n]$, $I_{S\to S'}= b_{S'^c}-b_{\{S\cup S'\}^c}.$
\label{LemInnov}
\end{lem}
\begin{prf}
Proof omitted due to lack of space.
\end{prf}

Since the occupancy vector provides the innovativeness of the contents of each node w.r.t the remaining nodes, we need to be able to track the dynamics of the occupancy vector for successful transmissions on links to complete the system modeling. To do so, let superscripts $-$ and $+$ denote the status of a system parameter before and after a successful packet transmission on a link. The following results derive the rules for updating the occupancy vector when successful transmissions occur. Throughout these results, we denote \emph{whp}/\emph{wlp} to qualify an event if its probability of occurrence can be made arbitrarily close to unity/zero by increasing the field size alone.

\begin{lem}
(\emph{Source-to-Relay}) The update rules when a relay $i$ successfully receives a packet from $s$ are as follows \emph{whp}.
\begin{itemize}
\item If $i\in S\subseteq [n]$ and $b_{\{i\}}<m_i$, then $b_S^+=b_S^-+1$.
\item If $i\notin S\subseteq [n]$, $b_{\{i\}}<m_i$ and $I_{\{i\}\to S^c\setminus \{i\}}^{-}=m_i$, then $b_S^+=b_S^-+1$.
\item Otherwise, $b_S^+=b_S^-$.
\end{itemize}
\label{LemSR}
\end{lem}

\begin{prf}
Proof omitted due to lack of space.
\end{prf}

\begin{lem}
(\emph{Relay-to-Relay}) The update rules when relay $j$ successfully receives a packet from relay $i$ are as follows \emph{whp}.
\begin{itemize}
\item If $i\in S\subseteq [n]$, $j\in S^c$, $I_{\{j\}\to S^c\setminus \{j\}}^{-}<m_j$ and $I_{\{i\}\to S^c}^{-}>0$, then $b_S^+=b_S^--1$.
\item Otherwise, $b_S^+=b_S^-$.
\end{itemize}
\label{LemRR}
\end{lem}

\begin{prf}
See Appendix~\ref{prfRR}.
\end{prf}

\begin{lem}
(\emph{Relay-to-Destination}) The update rules when $d$ successfully receives a packet from relay $j$ are as follows \emph{whp}.
\begin{itemize}
\item If $i\in S\subseteq [n]$ and $I_{\{i\}\to S^c}^{-}>0$, then $b_S^+=b_S^--1$.
\item Otherwise, $b_S^+=b_S^-$.
\end{itemize}
\label{LemRD}
\end{lem}

\begin{prf}
Proof omitted due to lack of space.
\end{prf}

On the whole, an update of buffer occupancy occurs only when the delivered packet is innovative for the receiving node and the buffer of the receiving node is not full. Next, we describe how the state update rules could be utilized to obtain the throughput of a network. Let $\E^*=(\e_1,\ldots, \e_{|\E|})$ be an ordering of the edge set $\E$, and let $l(t)\in\{0,1\}^{|\E|}$ represent the realization of the channels at time $t$. That is $l_i(t)=1$ if the $i^\textrm{th}$ edge $\e_i$ in $\E^*$ does not erase the transmitted packet during the epoch $t$. Then, given the occupancy vector $\{b_S(t)\}_{S\subseteq[n]}$ and the channel realization $l(t)$, the occupancy vector $\{b_S(t+1)\}_{S\subseteq[n]}$ can be determined using the state update rules presented in Lemmas~\ref{LemSR},~\ref{LemRR},~\ref{LemRD}.

Further, the state transition probability matrix $\T$ for the corresponding Markov chain can be identified as follows. Also, let $T_{\e}$ be the state transition matrix given a successful packet transmission on the link $\e$. For any $\e\in\E$, $T_{\e}$ can be determined using Lemmas~\ref{LemSR},~\ref{LemRR},~\ref{LemRD}. Therefore,

\begin{equation}
\T=\sum_{l\in\{0,1\}^{|\E|}} \Big(\prod_{j:l_j=0}\ve_{\e_j}\Big)\Big(\prod_{i:l_i=1}\ol\ve_{\e_i}T_{\e_i}\Big).
\end{equation}

This Markov chain can be proved to be \emph{irreducible, aperiodic}, and \emph{ergodic}~\cite{Feller1957, Vellambi2010}. Therefore, it possesses a unique steady-state probability distribution. Moreover, due to ergodicity, the time averages are equivalent to the statistical averages. Therefore, the throughput capacity $\mcc(\ora{G})$ can be determined using the steady state probability of the event that the network is in a state wherein the nodes possessing a link to the destination have innovative packets as follows.

\begin{equation}
\mcc(\ora{G})=\sum_{l\in\{0,1\}^{|\E|},\{b_S(t)\}}\mathfrak{N}(l,\{b_S(t)\})\cdot \Pr\Big(\{b_S(t)\}\Big),
\end{equation}
where $\mathfrak{N}(l,\{b_S(t)\})$ represents the number of successfully transmitted packets when state $\{b_S(t)\}$ and channel realization $l$ occur together.

\section{State Size Reduction in a Class of Networks}
\label{FB-sec3}

In Section~\ref{FB-sec2}, we observed that the number of state variables that we need to track at each time epoch is $2^n-1$ since $b_S$, the innovativeness of every subset of relay nodes w.r.t. its complement, must be considered. In this section, we show that all innovativeness terms need not be tracked to completely define the state of the system. This is a consequence of the intuition gained in line networks~\cite{Vellambi2010}. In line networks, we need to only track $I_{i\to S}$, where $S=\{i+1,\ldots,n\}$, i.e., all those intermediate nodes that are farther from the source hop-distance-wise. Equivalently, for line networks, it suffices that we track $b_S$ for $S=\{1,\cdots,i\}$ for $i\in[n]$. Extending that intuition, define $\mathscr{A}\triangleq \{S\subseteq [n]: \textrm{Every } j\in S^c \textrm{ has a path in $S^c$ to } d\}$ as illustrated in Fig.~\ref{FB-HNet}. 
\begin{figure}[htbp!]
\centering
\includegraphics[width=2.25in]{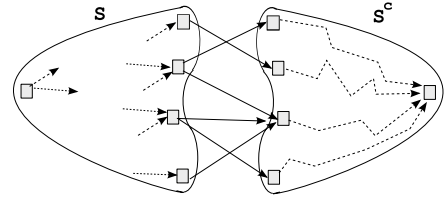}
\caption{Illustration of a set $S$ in $\mathscr{A}$.}\label{FB-HNet}
\end{figure}
\begin{figure}[htbp!]
\centering
\vspace{-.15in}
\includegraphics[width=1.7in,angle=-90]{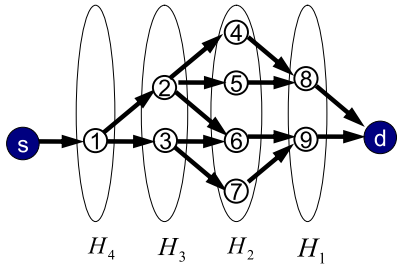}
\vspace{-.175in}
\caption{An example of a directed acyclic network in $\mathcal{N}$.}\label{FB-HNet}
\vspace{-.05in}
\end{figure}
Consider a partition of the set of relay nodes into types $\{H_1, H_2,\ldots\}$, where a relay node $v$ belongs to $H_k$ if the shortest hop-distance from $v$ to the destination $d$ is $k$, and $H_0\triangleq\{d\}$.  Define a class of networks $\mathcal{N}$ where every link starts at some node in $H_i$ for some $i$ and ends at some node in $H_{i-1}$. Figure~\ref{FB-HNet} illustrates a network from this class.
This structure enables us to track significantly lesser number of innovativeness components using the following result, which shows that tracking the occupancy for sets in $\mathscr{A}$ suffices to define the system completely.
\begin{thm}
For any directed acyclic network in $\mathscr{N}$, we need only track $b_S$ for $S\in\mathscr{A}$.
\end{thm}
\begin{proof}
Proof omitted due to lack of space.
\end{proof}

\section{Simulation Results}
\label{FB-sec4}
\vspace{-.025 in} 

In this section, we present the results of our performance modeling framework using state update rules in comparison with an actual packetized implementation of RLNC, and will show that our framework accurately models the buffer dynamics of the network.

We consider Network $1$ and Network $2$ shown in Fig.~\ref{FB-Net} to compare the results of our simulations.
\begin{figure}[htbp!]
\vspace{-.1in}
\subfigure[Network $1$.]{
\includegraphics[width=0.65in,angle=-90]{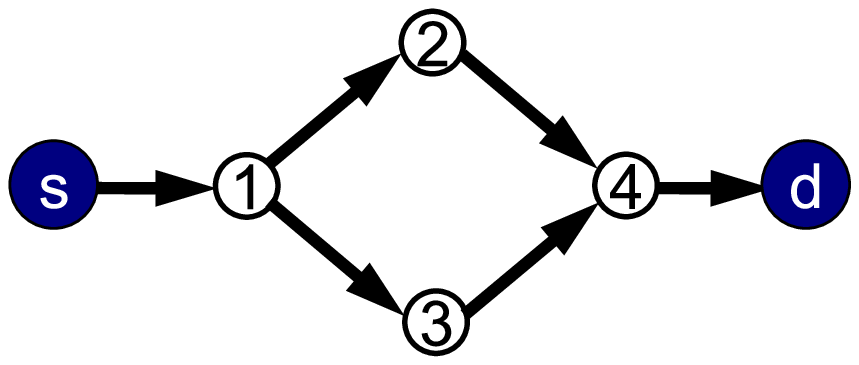}}
\subfigure[Network $2$.]{\includegraphics[width=0.65in,angle=-90]{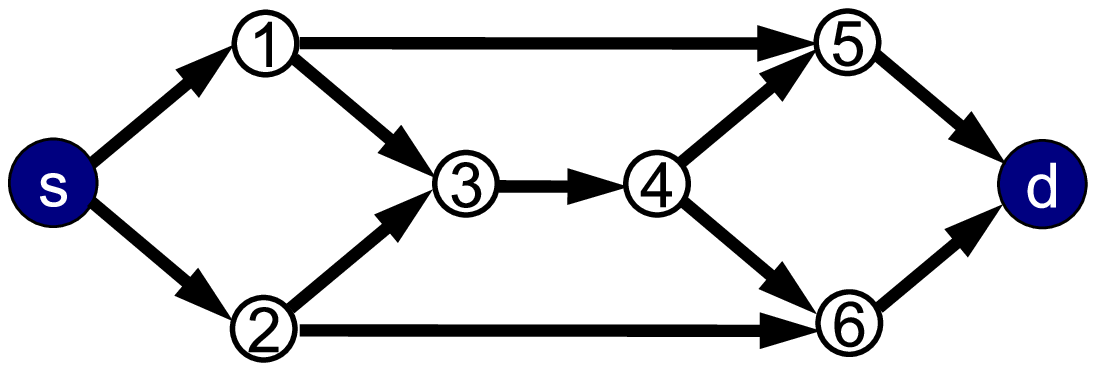}}
\vspace{-.05in}
\caption{Networks Considered for simulation}
\vspace{-.1in}
\label{FB-Net}
\end{figure}
In Network $1$, the edges have erasure probabilities $\ve_{(s,1)}=0.1$, $\ve_{(1,2)}=0.6$, $\ve_{(1,3)}=0.5$, $\ve_{(2,4)}=0.4$, $\ve_{(3,4)}=0.5$, and $\ve_{(4,d)}=0.1$. In Network $2$, all the edges have $\ve=0.5$ except the edges $\{(s,1), (s,2), (5,d), (6,d)\}$ for which $\ve=0.25$.
All the intermediate nodes are assumed to have the same buffer size. In order to measure the exact performance parameters of this network, a block of size $k=10^5$ packets is sent from the source to the destination.
\begin{figure}[htbp!]
\centering
\psfrag{Yaxis}{\small{\hspace{-3mm}Throughput (\emph{packets/epoch})}}
\psfrag{Xaxis}{\small{Buffer size (\emph{packets})}}
\psfrag{SimulationSimulationSimulationSimulation1}{\footnotesize{Sim. of actual packetized RLNC}}
\psfrag{SimulationSimulationSimulationSimulationrrrrrr}{\footnotesize{Sim. using state update rules}}
\includegraphics[width=3.4in, height = 2in]{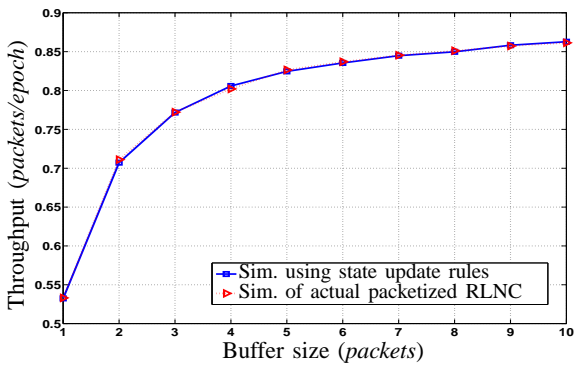}
\vspace{-.05in}
\caption{Throughput of Network $1$ for different buffer sizes.}\label{Net1_Results}
\vspace{-.05in}
\end{figure}
\begin{figure}[htbp!]
\centering
\psfrag{Yaxis}{\small{\hspace{-3mm}Throughput (\emph{packets/epoch})}}
\psfrag{Xaxis}{\small{Buffer size (\emph{packets})}}
\psfrag{SimulationSimulationSimulationSimulation1}{\footnotesize{Sim. of actual packetized RLNC}}
\psfrag{SimulationSimulationSimulationSimulationrrrrrr}{\footnotesize{Sim. using state update rules}}
\includegraphics[width=3.4in, height = 2in]{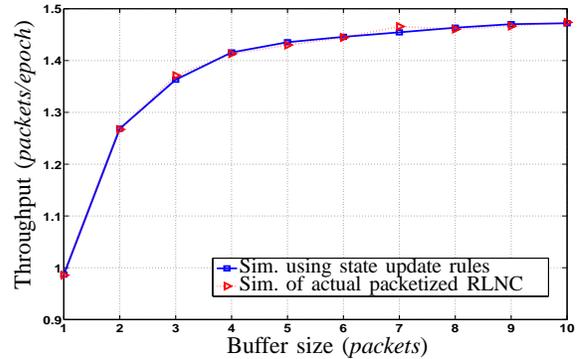}
\vspace{-.05in}
\caption{Throughput of Network $2$ for different buffer sizes.}\label{Net2_Results}
\vspace{-.05in}
\end{figure}
Fig.~\ref{Net1_Results} and Fig.~\ref{Net2_Results} present the variations of the throughput measured by actual simulation of RLNC and the throughput measured by simulation based on the state update rules developed in our work versus the buffer size. As it can be observed, our model is very close to the actual simulation results. Further, it confirms the optimality of RLNC for the infinite buffer setting as the curve approaches to the min-cut capacity for both networks. It is notable that the emulation of the RLNC using the derived state update rules takes significantly lesser time than the exact simulation of the RLNC scheme.
\begin{table}[h!]
\begin{center}
\caption{Number of active states vs. buffer size in Network 1.}$\,$\\
{
\renewcommand{\arraystretch}{1}
\small
\begin{tabular}{||c|c|c||} \hline
Buffer Size & No. of Active States & Upper Bound $(m+1)^{15}$ \\
  \hline\hline
  1  &   44   &  32768  \\\hline
  2  &   600  &  14348907 \\\hline
  3  &   4358 &  1073741824 \\\hline
\end{tabular}
}
\label{FB-Net1Table}
\end{center}
\vspace{-.05in}
\end{table}
Table~\ref{FB-Net1Table} compares the number of states actually visited (identified by simulations) and a crude upper bound on the number of states in the Markov chain model. For Network 1, the number of state variables is $2^4-1=15$, and a provable upper bound for the number of states is $(m+1)^{15}$, where $m$ is the buffer size of each intermediate node. However, it is noticed from simulations that the number of states that is actually realized is much lesser than the bound. This observation signals suggests that a closer look at the Markov chain to reduce its size can simplify the model, thereby rendering it more easily tractable.

\section{Conclusion and Future Work}\label{FB-sec5}
We have derived a novel notion of buffer occupancy for RLNC in wired finite-buffer networks. Using this notion, we developed a Markov-chain-based framework that can identify the throughput offered by RLNC using Monte Carlo simulations. This framework offers significant computational benefits over a complete simulation of RLNC. Though the size of the Markov chain is exponential, simulations suggest that a very small portion of the state space is actually visited in reality. A closer look at the state space and a thorough analysis to reduce the state space needs to be performed to eventually derive analytical throughput estimates.
\bibliographystyle{ieeetr}
\bibliography{ITWRef}

\appendices
\section{Proof of Lemma~\ref{LemRR}}\label{prfRR}
From Definition~\ref{OccDef} it is clear that if $i,j\in S$, then $b_S^+=b_S^-$. The same applies when $i,j\in S^c$. For the case $i\in S^c, j\in S$, the update rule is $b_S^+=b_S^-$ and the proof is similar to the one presented for the case $i\in S, j\in S^c$, which is as follows.

Hence, here we only assume $i\in S, j\in S^c$. Let $\mca^-=\{A_1^-,A_2^-,\ldots ,A_{m_i}^-\}$, $\mcb^-=\{B_1^-,B_2^-,\ldots ,B_{|\mcb^-|}^-\}$, $\mcc^-=\{C_1^-,C_2^-,\ldots ,C_{m_j}^-\}$ and $\mcd^-=\{D_1^-,D_2^-,\ldots ,D_{|\mcd^-|}^-\}$ be the buffer contents of relay $i$, relays $S\setminus\{i\}$, relay $j$, and relays $S^c\setminus\{j\}$ before packet transmission, respectively. Suppose packet $E=\sum_{l=1}^{m_i}\alpha_lA_l^-$ successfully transfers from relay $i$ to relay $j$. Then, for any $S\subseteq [n]$, We will have $\mca^+=\mca^-$, $\mcb^+=\mcb^-$, $\mcd^+=\mcd^-$, and $\mcc^+=\{C_1^-+\beta_1E,C_2^-+\beta_2E,\ldots ,C_{m_j}^-+\beta_{m_j}E\}$. Note that the coefficients $\alpha_l$ and $\beta_k$ are chosen randomly from $\mathbb{F}_q$. Let $\mcg^-=\text{span}\{\mca^-\}\cap \text{span}\{\mcc^-\cup \mcd^-\}$. We consider two cases:
\begin{itemize}
\item \textbf{Case 1}: Suppose there exists $\lambda_l,\theta_k$ such that $\lambda_l\neq 0$ for at least one $l$ and $\sum_l\lambda_lC_l^-+\sum_k\theta_kD_k^-=0$. Hence,
\begin{align}
\sum_l\lambda_lC_l^++\sum_k\theta_kD_k^-=
(\sum_l\lambda_l\beta_l)E &\in& \text{span}\{\mcc^+\cup \mcd^+\}\nonumber
\end{align}
Therefore, $E\in \text{span}\{\mcc^+\cup \mcd^+\}$ whp. Further, if $\mcg^- \neq \text{span}\{\mca^-\}$, then $E\notin \mcg^-$ whp, and $\text{span}\{\mcc^+\cup \mcd^+\}=\text{span}\{\mcc^-\cup \mcd^-\cup \{E\}\}$.
Hence,
\begin{equation}
\vspace{-.1in}
\begin{array}{lll}
b_S^+ &=& \dim(\text{span}\{\mca^-\cup \mcb^-\})\\
&&-\dim(\text{span}\{\mca^-\cup \mcb^-\}\cap \text{span}\{\mcc^-\cup \mcd^-\cup \{E\}\})\\
&=&b_S^--1
\end{array}\nonumber
\end{equation}

Note that $\mcg^- \neq \text{span}\{\mca^-\}$ $\Leftrightarrow$ $I_{\{i\}\to S^c}^{-}>0$, and the existence of such $\lambda_l,\theta_k$ $\Leftrightarrow$ $I_{\{j\}\to S^c\setminus \{j\}}^{-}<m_j$.

On the other hand, if $\mcg^- = \text{span}\{\mca^-\}$, then $E\in \mcg^-$ and since $\mcg^+=\mcg^-$, we will have $b_S^+=b_S^-$.

\item \textbf{Case 2}: Suppose no such $\lambda_l,\theta_k$ as in Case 1 exist.
Let $\mcf^-=\{F_i^-,i\in[|\mcf^-|]\}$ be a basis for $\text{span}\{\mca^-\cup \mcb^-\}\cap \text{span}\{\mcc^-\cup \mcd^-\}$ with $F_l^- = \sum_k\gamma_{lk}C_k^- + \sum_{k'}\mu_{lk'}D_{k'}^-$. Also, let $\mcf^+=\{F_1^+,F_2^+,\ldots,F_{|\mcf^-|}^+\}$, where
 \begin{equation}
\vspace{-.1in}
 F_l^+ \hspace{-0.75mm}= \hspace{-0.75mm}F_l^-\hspace{-0.75mm}+\hspace{-0.75mm}(\sum_k\gamma_{lk}\beta_k)E, \,\,l\in\{1,2,\ldots,|\mcf^-|\}.\label{basis+-}
 \end{equation}
Note that $F_l^+\in \text{span}\{\mca^+\cup \mcb^+\}\cap \text{span}\{\mcc^+\cup \mcd^+\}$.

Suppose $x\in \text{span}\{\mca^+\cup \mcb^+\}\cap \text{span}\{\mcc^+\cup \mcd^+\}$, then there exists representations of $x$ as follows.
\begin{equation}
\vspace{-.1in}
x = \sum_k\eta_kA_k^- + \sum_{k'}\delta_{k'}B_{k'}^-
= \sum_l\xi_l(C_l^-+\beta_lE) + \sum_{l'}\zeta_{l'}D_{l'}^-
\nonumber
\end{equation}
Therefore, we have
\begin{align}
 x-(\sum_l\xi_l\beta_l)E &\in \text{span}\{\mca^-\cup \mcb^-\}\cap \text{span}\{\mcc^-\cup \mcd^-\}\nonumber\\
\Rightarrow x-(\sum_l\xi_l\beta_l)E &= \sum_l\tau_lF_l^-= \sum_l\tau_l(F_l^+-\sum_k\gamma_{lk}\beta_k)E)\nonumber
\end{align}
Therefore,
\begin{align}
\hspace{-3mm}x\hspace{-0.75mm}-\hspace{-0.75mm}\sum_l\tau_lF_l^+ \hspace{-0.75mm}=\hspace{-0.75mm} \Big( \sum_l\xi_l\beta_l \hspace{-0.5mm}-\hspace{-0.5mm} \sum_{k,l} \tau_l\gamma_{lk}\beta_k \Big)E\hspace{-0.5mm}=\hspace{-0.5mm} \Phi(x)E
\label{phix}
\end{align}

We consider two cases here.

\textbf{Sub-case 2a}: First, suppose that $\Phi(x)=0$ for all $x\in \text{span}\{\mca^+\cup \mcb^+\}\cap \text{span}\{\mcc^+\cup \mcd^+\}$. Hence, $\text{span}\{\mcf^+\} = \text{span}\{\mca^+\cup \mcb^+\}\cap \text{span}\{\mcc^+\cup \mcd^+\}$.
Next, we prove that members of $\mcf^+$ are linearly independent. Suppose $\sum_l\omega_lF_l^+=0$, then by (\ref{basis+-}),
\begin{equation}
\sum_l\omega_lF_l^-=\Big(\sum_{l,k}\omega_l\lambda_{lk}\beta_k\Big)E
\label{WiFi}
\end{equation}
Here, if $\mcg^- \neq \text{span}\{\mca^-\}$, then $E\notin \mcf^-$ whp, and $\mcf^+$ are linearly independent, again whp. On the other hand, if $\mcg^- = \text{span}\{\mca^-\}$, then $E\in \mcf^-$ can be uniquely represented as a linear combination of $F_i^-$, $i\in[|\mcf^-|]$. Let $E=\sum_l\psi_lF_l^-$. Given a particular value of $(\omega_1,\cdots,\omega_{|\mcf^-|})\neq \mathbf{0}$, due to the randomness of the $\beta_k$'s, the probability that $\sum_l\omega_lF_l^+=0$ happens is equal to $\frac{1}{q-1}$ which can be made as small as required by choosing a large field size.

Thus, $\mcf^+$ are linearly independent in this case. Therefore,
\begin{equation}
\dim(\text{span}\{\mca^+\cup \mcb^+\}) = \dim(\text{span}\{F^+\})
= \dim(\text{span}\{F^-\}).
\nonumber
\end{equation}
Therefore, the update rule will be $b_S^+=b_S^-$.

\textbf{Sub-case 2b}: suppose that $\Phi(x)\neq 0$ for some $x\in \text{span}\{\mca^+\cup \mcb^+\}\cap \text{span}\{\mcc^+\cup \mcd^+\}$. Then, from (\ref{phix}), $E\in \text{span}\{\mca^+\cup \mcb^+\}\cap \text{span}\{\mcc^+\cup \mcd^+\}$. Now, if $\mcg^- = \text{span}\{\mca^-\}$, then $E\in \text{span}\{\mcc^-\cup \mcd^-\}$ which means that $\text{span}\{\mcc^+\cup \mcd^+\} = \text{span}\{\mcc^-\cup \mcd^-\}$. Thus, the update rule in this case is given by $b_S^+=b_S^-$.
On the other hand, if $\mcg^- \neq \text{span}\{\mca^-\}$, then $E\notin \text{span}\{\mcc^-\cup \mcd^-\}$. However, by (\ref{phix}), $E\in \text{span}\{\mcc^+\cup \mcd^+\}$. Hence, there exists a representation of $E$ as follows
\begin{equation}
E = \sum_l\pi_l(C_l^-+\beta_lE) + \sum_{l'}\varphi_{l'}D_{l'}^-
\label{eq4}
\end{equation}
\vspace{-.1in}
\begin{equation}
\Rightarrow\left(1-\sum_l\pi_l\beta_l\right)E = \sum_l\pi_lC_l^- + \sum_{l'}\varphi_{l'}D_{l'}^-.
\label{eq5}
\end{equation}
Given that $E\notin \text{span}\{\mcc^-\cup \mcd^-\}$, it follows from (\ref{eq5}) that $\sum_l\pi_l\beta_l=1$ which implies that
\begin{equation}
\sum_l\pi_lC_l^- + \sum_{l'}\varphi_{l'}D_{l'}^- = 0.
\label{eq6}
\end{equation}
However, in Case $2$, there cannot be an equation of the form (\ref{eq6}), unless we have $\pi_l=0$ for all $l$. Substituting $\pi_l=0$ in (\ref{eq4}) results in a contradiction. Thus, Sub-case 2b occurs wlp.\hfill$\blacksquare$
\end{itemize}

\end{document}